\documentclass[12pt]{amsart}
\usepackage[english]{babel}
\usepackage[utf8]{inputenc}
\usepackage{amsmath}
\usepackage{graphicx}
\usepackage{subfigure}
\usepackage{amssymb}
\usepackage{amsthm}
\usepackage{bm}
\usepackage{tikz-cd}
\usepackage{mathrsfs}
\usepackage[colorinlistoftodos]{todonotes}
\usepackage{enumitem}
\usepackage{yfonts}
\usepackage{dsfont}
\usepackage{bbm}
\usepackage{geometry}
\usepackage{setspace}
\usepackage{tikz}
\usepackage{pgfplots}
\usepackage{cite}
\pgfplotsset{compat=newest}
\usepgfplotslibrary{fillbetween}
\usetikzlibrary{positioning}
\usetikzlibrary{decorations.pathreplacing}
\usetikzlibrary{decorations,arrows}
\usetikzlibrary{decorations.markings}
\usetikzlibrary{patterns}
\usetikzlibrary{quotes,angles}
\usetikzlibrary{arrows}

\numberwithin{equation}{section}
\newtheorem{theorem}{Theorem}[section]
\newtheorem{lemma}[theorem]{Lemma}

\newtheorem{corollary}[theorem]{Corollary}

\newtheorem{remark}[theorem]{Remark}
\newtheorem{proposition}[theorem]{Proposition}
\newtheorem{assumption}[theorem]{Assumption}
\allowdisplaybreaks[4]

\title[Absence of a.c. spectrum for capacitance operators]
{Absence of Absolutely Continuous Spectrum for One-Dimensional Capacitance Operators with Critical Long-Range Interactions}

\author{Habib Ammari}
\address[H. Ammari]{Department of Mathematics, ETH Z\"{u}rich, R\"{a}mistrasse 101, CH-8092 Z\"{u}rich, Switzerland}
\email{habib.ammari@math.ethz.ch}

\author{Jiayu Qiu}
\address[J. Qiu]{Department of Mathematics, ETH Z\"{u}rich,
R\"{a}mistr. 101, CH-8092 Z\"{u}rich, Switzerland}
\email{jiayu.qiu@sam.math.ethz.ch}

\date{}

\begin{document}

\begin{abstract}
In this paper, we establish the absence of absolutely continuous spectrum for scalar one-dimensional Laurent operators with multiplicative disorder and critical long-range hoppings. In particular, the requirement of at least quartic off-diagonal decay rate in the existing theory is bypassed by exploiting the finite-difference bound of the convolution kernel, which is available even when the interaction decays only as $|n|^{-1}\log^{-\beta}|n|$ with $\beta>1$. We then apply the abstract main result to a concrete physical model, the capacitance operator in subwavelength physics.  For a chain of three-dimensional resonators with one resonator per period, we verify the required decay and finite-difference estimates and conclude the absence of absolutely continuous spectrum. These results leave open the distinction between pure point and singular continuous spectrum, but constitute a key step in proving Anderson localization in wave systems with critical long-range interactions.  


\end{abstract}

\maketitle



\tableofcontents

\section{Introduction}

Disorder can fundamentally change the spectral properties of a periodic wave system. In a periodic structure, the spectrum is decomposed by the Floquet-Bloch transform into a union of bands, which leads to its absolute continuous nature \cite{kuchment2016overview}. In contrast, random perturbations may destroy the absolutely continuous component of the spectrum, and under stronger assumptions the operator may only have a pure point spectrum with localized eigenvalues, as the main idea of Anderson localization; see \cite{stolz2011introduction} for an introduction. For one-dimensional lattice operators with short-range off-diagonal hopping terms, the mathematical theory of Anderson localization is well-developed; see, for example, the survey in~\cite{kirsch}. On the other hand, the long-range off-diagonal hopping rate introduces a fundamental difficulty, which is the main focus of this paper. We will consider a class of long-range one-dimensional lattice operators, whose off-diagonal decay rate falls outside the existing theory on Anderson localization, and study its spectrum under multiplicative identically and independently distributed (i.i.d.) perturbation.

The class of long-range lattice operators considered in this paper is motivated by a concrete model arising in the field of subwavelength physics, the capacitance operators; see \cite{lemoult.kaina.ea2013Wave,lemoult.kaina.ea2016Soda,leroy.devaud.ea2002air,liu.zhang.ea2000Locally,ma.sheng2016Acoustic} for unusual properties of subwavelength resonant systems and \cite{cbms} for their mathematical foundations. As shown in \cite{cbms, ammari.davies.ea2024Anderson,ammari.davies.ea2024Functional}, the capacitance operator describes the leading-order spectral and transport properties of subwavelength high-contrast resonator systems. Interestingly, when considering a periodic resonator chain embedded in $\mathbb{R}^3$, the associated capacitance operator $\mathcal{C}$ exhibits a long-range off-diagonal hopping rate $ \mathcal{C}_{0,n}\sim |n|^{-1}\log^{-2}|n|$; see \cite{ammari2026longrange}. This decay rate is critical in the sense that it is slower than $|n|^{-p}$ for any $p> 1$, while it is still absolutely summable. Importantly, this decay rate falls outside the scope of the established mathematical theory of Anderson localization, as discussed below.

The Anderson localization of one-dimensional lattice operators has been studied extensively; see \cite{aizenman.molchanov1993Localization,jaksicmolchanov1999,simonspencer1989,shi2023almost_periodic_long_range,qi2026transfer} and the references therein. In particular, under suitable assumptions on the Floquet symbol and the random potential, Jakšić and Molchanov \cite{jaksicmolchanov1999} prove the absence of absolutely continuous spectrum at arbitrarily weak disorder for off-diagonal hopping rate decaying strictly faster than \(|n|^{-4}\). Their theorem does not cover the critical decay of the capacitance operator considered here. Our contribution is to exploit an additional structure of this physical model: we find that the capacitance operator satisfies the following \textit{finite-difference estimates}
\[
 |\Delta^{(j)}\mathcal{C}_{0,n}|\sim
 \frac{1}{(1+|n|)^{j+1}\log^2(e+|n|)},
 \qquad j=0,1,2.
\]
With this finite-difference estimate, we are able to combine the sparse-gap construction in \cite{jaksicmolchanov1999} with a trace-class estimate of the spatial cut as in \cite{simonspencer1989}, by which we conclude the proof using the stability of absolutely continuous spectrum under trace-class perturbations; see Section \ref{sec_main_proof} for details. For general reference resonators, we prove the absence of absolute continuous spectrum within a local acoustic-window. In addition, when the reference resonator is reflectional symmetric, we are able to prove the absence of an absolutely continuous component for the whole spectrum; see Section \ref{sec_cap} for details.

The paper is organized as follows. Section~\ref{sec_main_result} states the abstract main result (Theorem \ref{thm_absence_ac}), which establishes the absence of absolutely continuous spectrum (upon a multiplicative i.i.d. perturbation) under certain hypotheses on the Floquet symbol, the off-diagonal decay rate, and the finite-difference estimates of the one-dimensional lattice operators. The proof of Theorem \ref{thm_absence_ac} is presented in Section \ref{sec_main_proof} using the sparse-gap construction, trace-class estimates of the spatial cut, and the Kato-Rosenblum theorem. Section~\ref{sec_cap} verifies  the decay and local Floquet hypotheses for one-resonator chains embedded in $\mathbb{R}^3$ and states the global extension under reflection symmetry. Appendix~\ref{app:reflection-monotonicity} proves the corresponding strict monotonicity by a half-cell Dirichlet principle. Section~\ref{sec:conclusion} discusses the conclusions and the remaining localization questions.

\section{Setup and main results} \label{sec_main_result}

Let $\mathcal{C}$ be a bounded periodic self-adjoint operator on the Hilbert space $H:=\ell^2(\mathbb{Z};\mathbb{C})$ with its convolution kernel denoted by $\mathcal{C}_{n,m}\in \mathbb{C}$ ($n,m\in\mathbb{Z}$). The periodicity indicates that $\mathcal{C}_{n+1,m+1}=\mathcal{C}_{n,m}$. The first assumption is on the off-diagonal decay rate of $\mathcal{C}$:
\begin{assumption} \label{asmp_decay}
There exists $\beta>1$ and $M>0$ such that for all $n,m\in\mathbb{Z}$ with $n\neq m$, it holds that
\begin{equation}
|\mathcal{C}_{n,m}|\leq \frac{M}{|n-m|\log^{\beta}(|n-m|+1)}
\end{equation}
\end{assumption}
Beyond off-diagonal decay, we pose the following assumption on the decay rate of finite differences of $\mathcal{C}_{n,m}$ up to the second order. More precisely, let $c_{n}:=\mathcal{C}_{0,n}$ and define
\begin{equation*}
\Delta^{(0)} c_{n}:= c_{n}, \quad \Delta c_{n}:=c_{n+1}-c_{n},\quad \Delta^{(j)}c_{n}:=\Delta (\Delta^{(j-1)}c_{n}) \quad (j> 1).
\end{equation*}
Then we assume the following. 
\begin{assumption} \label{asmp_finite-difference}
For each $j=1,2$, there exists $M_{j}>0$ such that
\begin{equation*}
|\Delta^{(j)}c_{n}|\leq \frac{M_{j}}{(1+|n|)^{j+1}\log^{\beta}(e+|n|)},\quad \forall n\in\mathbb{Z} .
\end{equation*}
\end{assumption}

The last assumption is on the regularity of the Floquet symbol,
\begin{equation*}
\hat{\mathcal{C}}(\kappa):=\sum_{n\in\mathbb{Z}}c_{n}e^{-in\kappa} \quad (\kappa\in [-\pi,\pi])
\end{equation*}
i.e., the Fourier transform of $\mathcal{C}$, which is well-defined since $c_{n}$ is absolutely summable by Assumption \ref{asmp_decay}. We assume:

\begin{assumption} \label{asmp_eigenvalue_monotone}
$\hat{\mathcal{C}}(\kappa)$ is a continuous real even function, twice continuously differentiable and strictly monotone on $(0,\pi)$.
\end{assumption}
Now we define the random operator considered here by introducing a multiplicative i.i.d. perturbation. Let $\{\omega_{n}\}_{n\in\mathbb{Z}}$ be a sequence of i.i.d. random variables with a common bounded density $\mu\in L^{\infty}(\mathbb{R})$ of a compact support whose interior contains the origin. We define the diagonal operator
\begin{equation*}
\mathcal{V}_{\omega,\epsilon}\in \mathcal{B}(H),\quad
(\mathcal{V}\psi)(n):=(1+\epsilon\omega_{n})\psi(n),
\end{equation*}
and the associated randomly perturbed Hamiltonian
\begin{equation*}
\mathcal{C}_{\omega,\epsilon}:=\mathcal{V}_{\omega,\epsilon}\mathcal{C}\mathcal{V}_{\omega,\epsilon}.
\end{equation*}

The main abstract result of this paper is the following theorem on the absence of an absolutely continuous spectrum for $\mathcal{C}_{\omega,\epsilon}$.
\begin{theorem} \label{thm_absence_ac}
Suppose that $\mathcal{C}$ satisfies Assumptions \ref{asmp_decay}, \ref{asmp_finite-difference}, and \ref{asmp_eigenvalue_monotone}. Then there exists $\epsilon_{0}>0$ such that for any $\epsilon\in (0,\epsilon_0)$, the absolute continuous spectrum of $\mathcal{C}_{\omega,\epsilon}$ is empty almost surely, i.e.,
\begin{equation*}
\text{Spec}_{a.c.}(\mathcal{C}_{\omega,\epsilon})=\emptyset,\quad \text{for a.e. $\omega$.}
\end{equation*}
\end{theorem}

When the monotonicity of Assumption \ref{asmp_eigenvalue_monotone} holds only near a neighborhood of some $\kappa\in (-\pi,\pi)$, Theorem \ref{thm_absence_ac} holds locally in the corresponding spectral window; see Proposition~\ref{prop:local-no-ac}.

\section{Proof of Theorem \ref{thm_absence_ac}}
\label{sec_main_proof}

The proof is carried out in two steps. First, at each energy $E\in \mathbb R\setminus\{0\}$, we construct a periodic configuration $\mathcal{V}^{(E)}$ such that the periodic operator $\mathcal{V}^{(E)}\mathcal{C}\mathcal{V}^{(E)}$ has a spectral gap at $E$. This construction is based on the Bragg scattering mechanism, whose essential idea is summarized in Section \ref{sec:block-supported-gaps-from-regularity}, as rigorously established in \cite{jaksicmolchanov1999}. Next, using the ergodicity of our system, we find arbitrarily accurate approximations on arbitrarily large finite windows of that configuration in almost every realization. Intuitively speaking, this indicates that our system is `gapped at infinity' and hence prevents any ballistic-type propagation; this illustrates physically our main result, i.e., the almost-surely vanishing absolute continuous spectrum. 

To rigorously formulate this intuition, we prove that, for almost every realization $\omega$, a smooth function of the corresponding Hamiltonian $\mathcal{V}_{\omega,\epsilon}\mathcal{C}\mathcal{V}_{\omega,\epsilon}$ differs from a direct sum of finite-dimensional subsystems only by a trace-class perturbation (the proof is based on the gap condition of $\mathcal{V}^{(E)}\mathcal{C}\mathcal{V}^{(E)}$ and the second-order difference estimate in Assumption \ref{asmp_finite-difference}); hence, using the Kato-Rosenblum theorem (cf. \cite[Chapter X, Theorem 4.4]{kato1995perturbation}), we conclude that the absolute continuous spectrum of $\mathcal{V}_{\omega,\epsilon}\mathcal{C}\mathcal{V}_{\omega,\epsilon}$ is empty.

In the sequel, for $a\in\mathbb Z$, we let $P_a:=\mathbbm 1_{(-\infty,a)}(X)$ (i.e., the projection to the half chain), where $X$ denotes the position operator on $\ell^2(\mathbb Z)$ defined as $(Xu)_n = n u_n, $ with the domain given by $\mathcal{D}(X)= \{ u \in \ell^2(\mathbb Z): \sum_{n \in \mathbb Z} n^2 |u_n|^2 < + \infty\}$. 

\subsection{Periodic sparse-gap opening}
\label{sec:block-supported-gaps-from-regularity}

We use the scalar sparse-gap theorem of Jak\v{s}i\'c and Molchanov \cite[Theorem~4.1]{jaksicmolchanov1999}.  In the terminology of their paper, let
\begin{equation*}
q_p(n):=\mathbbm 1_{p\mathbb Z}(n)
\end{equation*}
be the periodic sequence that equals one on the sites $p\mathbb{Z}$ and zero otherwise, and
\[
e_{k,p}:=\widehat{\mathcal C}(k\pi/p)
\]
be the Floquet symbol sampling at $k\pi/p$. The sparse-gap theorem states that, on every compact subinterval of the interior of $\sigma(\mathcal C)$, there are $\eta_0>0$ and $p_0$ such that, whenever
$0<\eta<\eta_0$, $p>p_0$, and $e_{k,p}$ belong to that subinterval,
\begin{equation}  \label{eq:JM-sparse-additive-gap}
 \sigma(\mathcal C+\eta q_p)
 \cap(e_{k,p},e_{k,p}+\delta_{\eta,k,p})=\varnothing
\end{equation}
for some $\delta_{\eta,k,p}>0$.

In this section, we convert the additive sparse gap into the periodic configuration allowed by our multiplicative model.  The assumption that zero is an interior point of the single-site support makes the conversion exact.

\begin{proposition}
\label{prop:block-supported-periodic-gaps}
Let $K_{\mathrm{supp}}:=\operatorname{supp}(\mu(t)\,dt)$ be the support of the random variable. There exists $\epsilon_0$ such that, for every $0<\epsilon<\epsilon_0$ and every $E\in\mathbb R\setminus\{0\}$ there are an integer $p\geq 1$ and a $p$-periodic sequence $\omega^{(E,\epsilon)}:\mathbb Z\to K_{\mathrm{supp}}$ such that
\begin{equation}
 E\notin\text{Spec}(\mathcal C_{\omega^{(E,\epsilon)},\epsilon}).
 \label{eq:scalar-supported-periodic-gap}
\end{equation}
\end{proposition}

Rather than reproducing the proof, we would like to emphasize the physical idea underlying \cite[Theorem~4.1]{jaksicmolchanov1999}, i.e., the Bragg scattering mechanism. Essentially, in the Bragg scattering process, a periodic potential of length $p$ mixes two counter-propagating Bloch modes with quasimomenta separated by $\frac{2\pi}{p}$. The upshot of such mixing is that, when the profile of the perturbation potential is appropriately designed, the two counter-propagating channels, which are referred to as resonant channels in the coupled two-mode description, give rise to two stationary modes \textit{with energy lifted from the original level $E$}. When the strength of the perturbation is equal to $\eta$, the energy fluctuation of the resonant modes is of the order 
$$\mathcal{O}(\eta).$$
On the other hand, since the unit period of the potential has length $p$, it samples Bloch modes with spacing roughly being $\mathcal{O}(p^{-1})$. This, together with the assumption that the counter-propagating Bloch modes have non-vanishing group velocity, implies that the non-resonant Bloch eigenvalues are separated from $E$ by the order of
\begin{equation*}
\mathcal{O}(p^{-1}).
\end{equation*}
Thus, we can construct a weak perturbation potential with a long period by controlling the two scales $p$ and $\eta$ (see the statement of \eqref{eq:JM-sparse-additive-gap}) so that the original energy level $E$ is not contained in the spectrum of the perturbed Hamiltonian. We refer the reader to \cite[Section 4]{jaksicmolchanov1999} for more details.

Note that Assumption \ref{asmp_eigenvalue_monotone} is precisely the hypothesis denoted by (H) in
\cite[Section 4]{jaksicmolchanov1999}, which is applied to prove Theorem 4.1 therein. Importantly, we point out that the absence of absolute continuous spectrum theorem in \cite{jaksicmolchanov1999}, as one of its main results, assumes an off-diagonal decay of Hamiltonians faster than $|n|^{-4}$, which obviously does not apply to the critical capacitance tail studied in this paper. In fact, we have only imported their sparse-gap theorem (which does not require a quartic off-diagonal decay rate) and overcome the requirement on the rapid decay rate by introducing and exploiting the finite-difference estimate stated in Assumption \ref{asmp_finite-difference}.

\begin{proof}[Proof of Proposition \ref{prop:block-supported-periodic-gaps}]

{\color{blue}Preparations.}
Since $K_{\mathrm{supp}}$ is compact, choose $\epsilon_0>0$ so that
\begin{equation} \label{eq:scalar-uniform-positive-weights}
 1+\epsilon t\geq\frac12 ,
 \quad \forall t\in K_{\mathrm{supp}},\, 0<\epsilon<\epsilon_0.
\end{equation}
Fix such an $\epsilon$ and $E\neq0$. If $E\notin\sigma(\mathcal C)$, the constant target $\omega^{(E,\epsilon)}\equiv0$ proves the assertion. We may therefore assume that $E\in\sigma(\mathcal C)$.

For $t$ near zero, we introduce the function
\begin{equation} \label{eq:scalar-fixed-energy-map}
 \Phi_{E,\epsilon}(t)
 :=E\left(1-(1+\epsilon t)^{-2}\right).
\end{equation}
Since $0\in\operatorname{int}K_{\mathrm{supp}}$ and $\Phi_{E,\epsilon}^{\prime}(0)=2E\epsilon\neq0$, there are $r>0$ and $a>0$ such that
\begin{equation} \label{eq:scalar-fixed-energy-image}
 (-r,r)\subset K_{\mathrm{supp}},
 \quad
 (-a,a)\subset\Phi_{E,\epsilon}((-r,r)).
\end{equation}
Choose a nonempty compact interval
\begin{equation}
 \mathcal J_E\Subset
 \sigma(\mathcal C)\cap(E-a/4,E+a/4).
 \label{eq:scalar-interior-energy-window}
\end{equation}
The existence of such an interval follows clearly by the band structure of $\sigma(\mathcal C)$. Let $\eta_0$ and $p_0$ be the constants supplied by \eqref{eq:JM-sparse-additive-gap} (i.e., the additive sparse-gap theorem) for $\mathcal J_E$. We further shrink the strength of potential $\eta$ by restricting that
\begin{equation}
 0<\eta<\min\{\eta_0,a/4\},
 \label{eq:scalar-comb-amplitude}
\end{equation}
and then $p>p_0$ and $k$ so that $e=e_{k,p}\in\mathcal J_E$; such a choice is possible because the symbol $\hat{\mathcal{C}}(\kappa)$ is continuous and the grids $\{k\pi/p:0\leq k\leq p\}$ become dense in $[0,\pi]$ as $p\to\infty$. 

{\color{blue}Construction of $\omega^{(E,\epsilon)}$.} With these ingredients, we now construct explicitly the periodic perturbation. Let $\delta=\delta_{\eta,k,p}$, which is also supplied by \eqref{eq:JM-sparse-additive-gap}. Select
\begin{equation}
 0<s<\min\{\delta,a/4\},
 \quad x:=E-e-s.
 \label{eq:scalar-gap-translation}
\end{equation}
Then $|x|<a/2$ and $|x+\eta|<3a/4$. By \eqref{eq:scalar-fixed-energy-image}, there are $t_0,t_1\in(-r,r)\subset K_{\mathrm{supp}}$ such that
\begin{equation}
 \Phi_{E,\epsilon}(t_0)=x,
 \quad
 \Phi_{E,\epsilon}(t_1)=x+\eta.
 \label{eq:scalar-two-supported-values}
\end{equation}
The desired $p$-periodic configuration is defined as
\begin{equation}
 \omega^{(E,\epsilon)}(n):=
 \begin{cases}
  t_1,&n\in p\mathbb Z,\\
  t_0,&n\notin p\mathbb Z.
 \end{cases}
 \label{eq:scalar-supported-sparse-target}
\end{equation}
and let $\mathcal V^{(E,\epsilon)}$ be the multiplication operator by
$1+\epsilon\omega^{(E,\epsilon)}(n)$. Then, using \eqref{eq:scalar-fixed-energy-map} and \eqref{eq:scalar-two-supported-values}, one can directly check that
\begin{equation*}
E\left(I-(\mathcal V^{(E,\epsilon)})^{-2}\right)
 =xI+\eta q_p.
\end{equation*}
Consequently,
\begin{equation} \label{eq:scalar-supported-gap-congruence}
\begin{aligned}
\mathcal V^{(E,\epsilon)}\mathcal C
 \mathcal V^{(E,\epsilon)}-E
 &=\mathcal V^{(E,\epsilon)}
 \left(\mathcal C-E(\mathcal V^{(E,\epsilon)})^{-2}\right)
 \mathcal V^{(E,\epsilon)} \\
 &=\mathcal V^{(E,\epsilon)}
 \left(\mathcal C+xI+\eta q_p-E\right)
 \mathcal V^{(E,\epsilon)}.
\end{aligned}
\end{equation}
Note that the spectrum of $\mathcal C+xI+\eta q_p$ has a gap $(e+x,e+\delta+x)=(E-s,E+\delta-s)$ which contains $E$. Hence, the middle factor in \eqref{eq:scalar-supported-gap-congruence} is invertible. Since the conjugation operator $V^{(E,\epsilon)}$ is also invertible by \eqref{eq:scalar-uniform-positive-weights}, we complete the proof of \eqref{eq:scalar-supported-periodic-gap}.
\end{proof}

\begin{remark} \label{rmk:extension_to_d_dim}
We note that the sparse-gap theorem is expected to be extended to the case where the Floquet symbol $\hat{\mathcal{C}}$ is matrix-valued (i.e., there are several resonators in a single block) and all branches of Bloch eigenvalues are regular near the energy level of interest $E$. In particular, the construction of a sparse gap in this case requires one to realize the Bragg scattering mechanism among $2N$ resonant channels, rather than only two. With this general result, the absence of absolute continuous spectrum theorem of this work can be generalized to periodic operators with matrix-valued symbols. We leave this interesting extension to future work.
\end{remark}
 
The periodic-gap construction is now complete, and the remaining proof of Theorem \ref{thm_absence_ac} no longer uses the regularity of the Floquet symbol. In the sequel, we show that the gap-opening profile obtained above appears as a two-sided right limit of almost every random configuration. Then, using such 'gapped at infinity' result, we obtain a sequence of cuts along the chain across which a localized spectral function has arbitrarily small trace-norm coupling.

\subsection{Trace-class spatial cut} \label{sec_trace_class_cut}

A periodic gap can act as a spectral barrier only if cutting the original long-range chain gives rise to a trace-class cost. This is justified in the following lemma using the decay estimate of the second-order finite difference of the kernel, as assumed in Assumption \ref{asmp_finite-difference}.

\begin{lemma}[Trace-class cost of a spatial cut]
\label{lem:trace-class-cut}
Let $P_0:=\mathbbm{1}_{n <0} (X)$ be the orthogonal projection onto the negative half-chain. Under Assumptions~\ref{asmp_decay} and \ref{asmp_finite-difference}, we have
\begin{equation}
 [P_0,\mathcal C]\in\mathfrak S_1(H),
 \label{eq:trace-class-cut}
\end{equation}
where $\mathfrak S_1(H)$ denotes the trace-class operators on $H$.
\end{lemma}

\begin{proof}
Let $Q_0:=I-P_0$ and decompose
\[
H=P_0H\oplus Q_0H
 =\ell^2(\mathbb Z_{<0})\oplus\ell^2(\mathbb Z_{\geq0}).
\]
Then
\begin{equation} \label{eq:commutator_block_form}
[P_0,\mathcal C]
 =
\begin{pmatrix}
0&P_0\mathcal C Q_0\\
-Q_0\mathcal C P_0&0
\end{pmatrix}.
\end{equation}
Under the identification of $\ell^2(\mathbb Z_{<0})$ with $\ell^2(\mathbb Z_{>0})$, the cross
block $Q_0\mathcal C P_0$ is represented by the Hankel matrix
\[
\Gamma:=(a_{j+k})_{j,k\geq0},\quad
a_n:=c_{-(n+1)},
\]
and the other cross block is $\Gamma^*$. Hence, it suffices to prove that $\Gamma$ is trace-class. Note that Assumptions \ref{asmp_decay} and \ref{asmp_finite-difference} imply $a_N\to0$ and $\Delta a_N\to0$. Consequently, we can apply the summation by parts twice and obtain
\begin{equation}
 a_n=\sum_{N=n}^{\infty}(N-n+1)\Delta^{(2)} a_N.
 \label{eq:twice-summed-difference}
\end{equation}
For $N\geq0$, define the matrix
\[
 (B_N)_{jk}:=(N-j-k+1)_+,
 \quad j,k\geq0.
\]
Then $\Gamma$ is factorized as
\begin{equation}
 \Gamma=\sum_{N=0}^{\infty}(\Delta^{(2)} a_N)B_N.
 \label{eq:Hankel-second-difference-expansion}
\end{equation}
In particular, the trace norm of $B_N$ is estimated elementarily as
\begin{equation}
 \|B_N\|_1\leq\frac{(N+1)(N+2)}2.
 \label{eq:BN-trace-norm}
\end{equation}
Using this estimate, we conclude that
\begin{align*}
 \|\Gamma\|_1
 &\leq\frac12\sum_{N=0}^{\infty}(N+1)(N+2)
 |\Delta^{(2)} a_N|\\
 &\leq C\sum_{N=1}^{\infty}
 \frac{1}{N\log^\beta(e+N)}<\infty.
\end{align*}
Thus, $[P_0,\mathcal C]$ is a trace class operator.
\end{proof}

\subsection{Almost-sure recurrence of periodic potentials}

Let $\Omega:=K_{\mathrm{supp}}^{\mathbb Z}$ be the configuration space of random variable, equipped with its product topology and product probability measure, and let $\tau$ be the left-shifting operator defined as $(\tau\omega)(n):=\omega(n+1)$.

The deterministic periodic gap-opening potential found in Proposition \ref{prop:block-supported-periodic-gaps} affects the random spectrum only if its profile reappears arbitrarily far to both sides of a typical random configuration. This recurrence actually emerges for almost every configuration, due to the ergodicity of our system.

\begin{lemma}[Two-sided recurrence of every finite pattern]
\label{lem:dense-two-sided-orbit}
There exists an event $\Omega_*$ of probability one such that, for every
$\omega\in\Omega_*$, both semi-orbits $\{\tau^n\omega:n\geq0\}$ and $\{\tau^{-n}\omega:n\geq0\}$ are dense in $\Omega$. In particular, both semi-orbits visit every nonempty relatively open finite-coordinate cylinder infinitely often.
\end{lemma}

\begin{proof}
First, we note that the product space $\Omega$ has a countable basis
of nonempty relatively open finite-coordinate cylinders defined by intervals with rational endpoints. Each such cylinder has a positive probability, since every nonempty relatively open subset of the single-site support has a positive measure. The Bernoulli shift and its inverse are ergodic. As a consequence, by the Birkhoff ergodic theorem applied to the indicator of each basic cylinder, we know that the two semi-orbits visit any given rational cylinder $\mathcal{S}$ with probability one infinitely often. Let us record this event as $\Omega_{\mathcal{S}}$. Finally, taking the countable intersection $\Omega_{*}:=\cap_{\mathcal{S}}\Omega_{\mathcal{S}}$ completes the proof.
\end{proof}

\subsection{Almost-decoupling cut}

The main result of this subsection states that, for smooth spectral functions $f(\mathcal C_{\omega,\epsilon})$ with $f$ supported near the energy level $E$ (recall that $E$ lies in the sparse-gap as proved in Proposition \ref{prop:block-supported-periodic-gaps}), the cost of spatial cut has a vanishing trace norm as the cut moves to infinity.

\begin{proposition} \label{prop:cost_spatial_cut}
Suppose that Assumptions~\ref{asmp_decay} and \ref{asmp_finite-difference} hold. Fix $\omega\in\Omega_*$, which is the event of probability one found in Lemma \ref{lem:dense-two-sided-orbit}, and $\epsilon\in (0,\epsilon_0)$. Let $\omega^{(E,\epsilon)}\in K_{\mathrm{supp}}^{\mathbb Z}$ be the periodic configuration found in Proposition \ref{prop:block-supported-periodic-gaps}, and $G$ be a bounded open interval containing $E$ with $\overline G\subset\rho(\mathcal C_{\omega^{(E,\epsilon)},\epsilon})$. Then there are sequences
$n_j^+\to+\infty$ and $n_j^-\to-\infty$ such that, for every
$f\in C_c^\infty(G;\mathbb R)$,
\begin{equation} \label{eq:good-cuts}
 \|P_{n_j^\pm}f(\mathcal C_{\omega,\epsilon})
 (I-P_{n_j^\pm})\|_1\longrightarrow0.
\end{equation}

\end{proposition}



\begin{proof}
{\color{blue}Step 1.} We first note that applying Lemma \ref{lem:dense-two-sided-orbit} to the given periodic configuration $\omega^{(E,\epsilon)}$ gives sequences $n_j^{\pm}\to\infty$ such that
\begin{equation} \label{eq:right-limit-sequences}
 \tau^{n_j^\pm}\omega\longrightarrow
 \omega^{(E,\epsilon)} \quad\text{pointwise}.
\end{equation}
Let $S$ be the bilateral shift, i.e., $(Su)(n):=u(n-1)$. Then, it holds that
\[
 S^nP_0S^{-n}=P_n,
 \qquad S^{-n}\mathcal CS^n=\mathcal C.
\]
Along either sequence in \eqref{eq:right-limit-sequences}, we define
\begin{equation} \label{eq:shifted-random-operators}
 V_j:=S^{-n_j^\pm}\mathcal V_{\omega,\epsilon}S^{n_j^\pm},
 \qquad
 A_j:=S^{-n_j^\pm}\mathcal C_{\omega,\epsilon}S^{n_j^\pm}
 =V_j\mathcal CV_j.
\end{equation}
Then \eqref{eq:right-limit-sequences} implies the following convergence in the strong operator topology
\begin{equation} \label{eq:strong-right-limit}
 V_j\longrightarrow\mathcal V^{(E,\epsilon)},
 \qquad
 A_j\longrightarrow
 C_{\omega^{(E,\epsilon)},\epsilon}=\mathcal V^{(E,\epsilon)}
 \mathcal C\mathcal V^{(E,\epsilon)} .
\end{equation}
{\color{blue}Step 2.} Recall the elementary fact that if $T$ is trace-class and the bounded operators $B_j,C_j$, together with their adjoints, converge strongly, then the product converges in trace norm
\begin{equation}
 B_jTC_j\longrightarrow BTC
 \quad\text{in }\mathfrak S_1.
 \label{eq:strong-trace-strong}
\end{equation}
Consequently, recalling that all diagonal multipliers commute with $P_0$, i.e.,
\[
 [P_0,A_j]=V_j[P_0,\mathcal C]V_j,
\]
and applying Lemma~\ref{lem:trace-class-cut} and \eqref{eq:strong-right-limit}, we conclude
\begin{equation}
 [P_0,A_j]\longrightarrow[P_0,A^{(E,\epsilon)}]
 \quad\text{in }\mathfrak S_1.
 \label{eq:commutator-trace-limit}
\end{equation}
{\color{blue}Step 3.} For $f\in C_c^\infty(G;\mathbb R)$, we denote $\widetilde f$ as an almost-analytic extension of $f$ (see \cite{zworski2012semiclassical} for the definition and a detailed discussion). The Helffer-Sj\"ostrand formula and the resolvent identity give
\begin{equation}
 [P_0,f(A_j)]
 =-\frac1\pi\int_{\mathbb C}\partial_{\overline{z}}\widetilde f(z)
 (A_j-z)^{-1}[P_0,A_j](A_j-z)^{-1}\,d^2z.
 \label{eq:HS-commutator}
\end{equation}
For every nonreal $z$, the integrand converges in the trace norm by strong resolvent convergence and \eqref{eq:commutator-trace-limit}. Its trace norm is bounded by
\[
\|(A_j-z)^{-1}[P_0,A_j](A_j-z)^{-1}\|_{1}\leq 
 |\operatorname{Im}z|^{-2}
 \sup_j\|[P_0,A_j]\|_1 .
\]
On the other hand, the almost-analytic extension satisfies
\begin{equation*}
|\partial_{\overline{z}}\widetilde f(z)|\leq C_{p}|\operatorname{Im}z|^{p}
\end{equation*}
for any $p\geq 2$ and $C_{p}>0$ depending only on $p$. Thus, the dominated convergence theorem yields
\begin{equation}
 [P_0,f(A_j)]\longrightarrow
 [P_0,f(A^{(E,\epsilon)})]=0
 \quad\text{in }\mathfrak S_1,
 \label{eq:smoothed-cut-limit}
\end{equation}
because $f(A^{(E,\epsilon)})=0$. Finally, for any orthogonal projection $P$ and bounded self-adjoint operator $F$, it holds that
\begin{equation}
 \|[P,F]\|_1=2\|PF(1-P)\|_1.
 \label{eq:commutator-cross-block}
\end{equation}
This, together with \eqref{eq:smoothed-cut-limit} and \eqref{eq:shifted-random-operators}, leads to
\begin{equation}
 \|P_{n_j^\pm}f(\mathcal C_{\omega,\epsilon})
 (1-P_{n_j^\pm})\|_1\longrightarrow0
\end{equation}
at cuts tending to both ends of the chain, which concludes the proof of \eqref{eq:good-cuts}.
\end{proof}

\subsection{Trace-class block decoupling}
\label{sec:trace-class-block-decoupling}

Finally, we apply Proposition \ref{prop:cost_spatial_cut} to prove Theorem \ref{thm_absence_ac}. The essential idea is that, by choosing the almost-decoupling cuts sparsely, we obtain a direct sum of finite-dimensional self-adjoint operators and a summable trace-class error. Then the Kato-Rosenblum theorem indicates that the Hamiltonian under consideration $\mathcal{C}_{\omega,\epsilon}$ has the same absolute continuous spectrum as the sum of finite-dimensional self-adjoint operators, which is therefore empty.

To be precise, we first extract an increasing bi-infinite sequence $(a_k)_{k\in\mathbb Z}$ from
the cuts in \eqref{eq:good-cuts}, tending to the corresponding infinity as
$k\to\pm\infty$, such that
\begin{equation}
 \sum_{k\in\mathbb Z} \|P_{a_k}f(\mathcal C_{\omega,\epsilon})(1-P_{a_k})\|_1<\infty. \label{eq:summable-cuts}
\end{equation}
Define
\[
 F:=f(\mathcal C_{\omega,\epsilon}),
 \qquad Q_k:=P_{a_{k+1}}-P_{a_k},
 \qquad D:=\bigoplus_{k\in\mathbb Z}Q_kFQ_k.
\]
Note that every $Q_k$ has a finite rank. Moreover, the series
\begin{equation}
 T=\sum_{k\in\mathbb Z}Q_kF(1-P_{a_{k+1}})
 \label{eq:upper-triangular-series}
\end{equation}
converges absolutely in the trace norm by \eqref{eq:summable-cuts}. On the other hand, we note that every strictly upper off-diagonal block of $F$ occurs exactly once in $T$, whence
\begin{equation}
 F-D=T+T^*\in\mathfrak S_1.
 \label{eq:trace-class-block-decoupling}
\end{equation}
Note that the operator $D$ is a direct sum of finite-dimensional self-adjoint matrices, so it has a pure point spectral type. The Kato-Rosenblum theorem (cf. \cite[Chapter X, Theorem 4.4]{kato1995perturbation}) indicates that the absolutely continuous parts
of $F$ and $D$ are unitarily equivalent and hence implies
\begin{equation}
 P_{\mathrm{ac}}(f(\mathcal C_{\omega,\epsilon}))=P_{\mathrm{ac}}(F)=P_{\mathrm{ac}}(D)=0.
 \label{eq:no-ac-smooth-function}
\end{equation}
Finally, we translate the spectrum of $f(\mathcal C_{\omega,\epsilon})$ to the original operator $\mathcal C_{\omega,\epsilon}$ using the standard spectral calculus. Choose a compact interval $J\Subset G$ containing $E$ in its interior and select $f$ to be affine with nonzero slope near $J$. By this construction, one can verify that a nonzero absolutely continuous spectral measure of $\mathcal C_{\omega,\epsilon}$ supported in $J$ implies that $f(\mathcal C_{\omega,\epsilon})$ also has such a component, which obviously contradicts \eqref{eq:no-ac-smooth-function}. Thus,
\[
 \mathbbm 1_J(\mathcal C_{\omega,\epsilon})
 P_{\mathrm{ac}}(\mathcal C_{\omega,\epsilon})=0.
\]
Moreover, we note that the event $\Omega_*$ was chosen independently of both $E$ and $\epsilon$. The interiors of the resulting intervals $J$ cover $\mathbb R\setminus\{0\}$ and admit a countable subcover. Hence,
\[
 \mathbbm 1_{\mathbb R\setminus\{0\}}
 (\mathcal C_{\omega,\epsilon})
 P_{\mathrm{ac}}(\mathcal C_{\omega,\epsilon})=0.
\]
Finally, note that an absolutely continuous spectral measure has no atom at zero, so its restriction to $\{0\}$ also vanishes. Therefore, we conclude that $P_{\mathrm{ac}}(\mathcal C_{\omega,\epsilon})=0$, which completes the proof of Theorem~\ref{thm_absence_ac}.

\subsection{A local absence result}
\label{sec:local-absence}

In this section, we present a local version of Theorem \ref{thm_absence_ac}, whose assumption on the global monotonicity of the Floquet symbol is weakened to hold only locally. In fact, this is again motivated by our application on the capacitance operator, in which the global monotonicity holds only under additional geometric assumption; see Section \ref{sec_cap} for details.

\begin{proposition}[Local absence of absolutely continuous spectrum]
\label{prop:local-no-ac}
Suppose Assumptions~\ref{asmp_decay} and \ref{asmp_finite-difference} hold. Moreover, assume that $h:=\widehat{\mathcal C}$ is real, even, and continuous on $[-\pi,\pi]$. Let $U\subset\mathbb R\setminus\{0\}$ be an open interval and assume that for every $E\in U$,
\[
 h^{-1}(\{E\})=\{-\theta(E),\theta(E)\},\qquad 0<\theta(E)<\pi,
\]
where $h$ is $C^2$ near $\pm\theta(E)$ and $h'(\theta(E))\neq0$. There are $\epsilon_0>0$ and a probability-one event $\Omega_*$ such that, for all $\omega\in\Omega_*$ and $0<\epsilon<\epsilon_0$,
\begin{equation}
 \mathbbm 1_U(\mathcal C_{\omega,\epsilon})
 P_{\mathrm{ac}}(\mathcal C_{\omega,\epsilon})=0.
 \label{eq:local-no-ac}
\end{equation}
\end{proposition}

In fact, the construction of the sparse-gap configuration in \cite{jaksicmolchanov1999} is purely local, which only requires the local monotonicity of the Floquet symbol as we assume in Proposition \ref{prop:local-no-ac}; we refer the readers to \cite[Section 4]{jaksicmolchanov1999} for details. Then such a sparse-gap construction is directly transferred to our multiplicative model using the same argument as in Proposition \ref{prop:block-supported-periodic-gaps}. After that, the proof of Proposition \ref{prop:local-no-ac} follows the same lines as in Sections \ref{sec_trace_class_cut}-\ref{sec:trace-class-block-decoupling}, which are skipped here.

\section{One-dimensional-in-three-dimensional capacitance operator}
\label{sec_cap}

\subsection{Definition of the capacitance operator}

In this section, we introduce the definition of the capacitance operator of a chain of high-contrast resonators in three dimensions; we refer the reader to \cite{ammari2026longrange} for a detailed discussion of its physical origin. Fix $a>0$ and introduce the unit cell
\[
 Y_c:=(-a/2,a/2)\times\mathbb R^2.
\]
Let $D\Subset Y_c$ be a connected domain with boundary of class $C^{2,\eta}$ (i.e., the reference resonator) and denote its translations by
\[
 D_n:=D+nae_1,
 \quad
 \mathcal D:=\bigcup_{n\in\mathbb Z}D_n.
\]
The capacitance operator associated with the resonator chain $\mathcal{D}$ is defined via its Floquet symbol. To introduce it, we first note that, for $\alpha\in[-\pi/a,\pi/a]\setminus\{0\}$, the quasiperiodic single-layer operator (cf. \cite{cbms}) is an invertible operator on $L^2(\partial D)$ defined as
\begin{equation}
 \mathcal S_D^\alpha[\varphi](x)
 =-\frac1{4\pi}\int_{\partial D}
 \sum_{n\in\mathbb Z}
 \frac{e^{i\alpha na}}{|x-y-nae_1|}\varphi(y)\,d\sigma_y.
 \label{eq:quasiperiodic-single-layer}
\end{equation}
With this notion, the Floquet symbol of the capacitance operator is defined as
\begin{equation}
 \widehat{\mathcal C}^{\alpha}
 :=-\int_{\partial D}
 (\mathcal S_D^\alpha)^{-1}[\mathbbm 1_{\partial D}]\,d\sigma ,
 \label{eq:Floquet-capacitance}
\end{equation}
and the capacitance coefficients are given by
\begin{equation}
 c_n:=\frac{a}{2\pi}\int_{-\pi/a}^{\pi/a}
 \widehat{\mathcal C}^{\alpha}e^{-i\alpha na}\,d\alpha.
 \label{eq:capacitance-Fourier-coefficients}
\end{equation}
The capacitance operator $\mathcal C$ is defined as the convolution operator with kernel $\mathcal C_{n,m}:=c_{n-m}$. With the dimensionless variable $\theta=a\alpha$, write
\begin{equation}
 c(\theta)=\widehat{\mathcal C}^{\theta/a},
 \quad -\pi\leq\theta\leq\pi.
 \label{eq:dimensionless-capacitance-symbol}
\end{equation}
The layer-potential construction yields
\begin{equation}
 c\in C^\infty(( -\pi,0)\cup(0,\pi]),
 \quad c(-\theta)=c(\theta)\in\mathbb R,
 \quad c(0)=0.
 \label{eq:basic-symbol-properties}
\end{equation}
Moreover, $c(\theta)>0$ for $\theta\neq0$ by the Dirichlet principle.  These facts, as well as the small-quasimomentum asymptotic expansion used below, follow from the quasiperiodic layer-potential analysis in \cite{ammari2026longrange}.

\subsection{Verification of Assumptions~\ref{asmp_decay} and
\ref{asmp_finite-difference}}

\subsubsection{Finite-order estimates of the Floquet remainder}

Note that in \cite[Theorem 3.8]{ammari2026longrange} we have derived the leading-order asymptotics of the capacitance coefficients:
\begin{equation}
 |c_n|\leq\frac{C}{(1+|n|)\log^2(e+|n|)},
 \label{eq:capacitance-coefficient-bound}
\end{equation}
which immediately verifies Assumption \ref{asmp_decay} with $\beta=2$. In this section, we prove the estimates of finite differences in Assumption \ref{asmp_finite-difference}. This is achieved by extracting a complete-Bernstein singular core from the capacitance operator together with a regular remainder. For simplicity, we assume $a=|\partial D|=1$ throughout this section as in the general case only the relevant constants in the estimates will change. Define
\begin{equation} \label{eq:kappa-lambda}
 \kappa(\theta):=\frac1{4\pi}\log(2(1-\cos\theta)).
\end{equation}
Let $P_\parallel$ denote the orthogonal projection onto the constant functions on $\partial D$ and let $P_\perp:=I-P_\parallel$. Note that the quasiperiodic single-layer operator has the decomposition (cf. \cite{ammari2026longrange})
\begin{equation}
 \mathcal S_D^\theta
 =\kappa(\theta)P_\parallel+A+R_{\theta},
 \label{eq:single-layer-small-Floquet}
\end{equation}
where the $\theta$-independent operator $A\in\mathcal B(L^2(\partial D),H^1(\partial D))$ is defined by
\begin{equation}
 (A\varphi)(x)
 :=-\frac1{4\pi}\int_{\partial D}
 \left[
  \frac1{|x-y|}
  +\sum_{n\in\mathbb Z\setminus\{0\}}
   \left(
    \frac1{|x-y-ne_1|}-\frac1{|ne_1|}
   \right)
 \right]\varphi(y)\,d\sigma_y
 \label{eq:definition-static-layer-operator}
\end{equation}
and $R_{\theta}$ includes all remaining terms. The verification of Assumption \ref{asmp_finite-difference} is based on the following complete asymptotics of $R_{\theta}$, extending those in \cite[Proposition 3.5]{ammari2026longrange}.

\begin{lemma}[Finite-order Floquet remainder]
\label{lem:finite-order-remainder}
For every fixed integer $M\geq2$, there are $\theta_0\in(0,1)$ and constants $C_0,\ldots,C_M>0$ such that, for $0<|\theta|<\theta_0$,
\begin{equation} \label{eq:remainder-first-bounds}
\|R_{\theta}\|\leq C_0|\theta|\log(e/|\theta|),
\quad\|\partial_\theta R_{\theta}\|\leq
 C_1\log(e/|\theta|),
\end{equation}
\begin{equation} \label{eq:remainder-higher-bounds}
 \|\partial_\theta^m R_{\theta}\|\leq C_m|\theta|^{1-m},
 \quad 2\leq m\leq M.
\end{equation}
The operator norms are taken in $\mathcal B(L^2(\partial D),H^1(\partial D))$. The constants depend only on $M$ and geometry of $D$, not on $\theta$.
\end{lemma}

\begin{proof}
{\color{blue}Step 1.} We begin by presenting an explicit formula of the remainder. For $z=(z_1,z_\perp)\in\mathbb R\times\mathbb R^2$ and $n\in\mathbb Z\setminus\{0\}$, define
\[
 K_n(z):=\frac1{|z-ne_1|}-\frac1{|ne_1|}.
\]
Hence, the remainder is rewritten as
\begin{equation}
 (R_{\theta}\varphi)(x)
 =-\frac1{4\pi}\int_{\partial D}\varphi(y)
 \sum_{n\in\mathbb Z\setminus\{0\}}
 K_n(x-y)(e^{in\theta}-1)\,d\sigma_y.
 \label{eq:explicit-Floquet-remainder}
\end{equation}
Next, we decompose $R_{\theta}=R_+^\theta+R_-^\theta$, where
\begin{align}
 (R_+^\theta\varphi)(x)
 &:=-\frac1{4\pi}\int_{\partial D}\varphi(y)
   \sum_{n=1}^\infty K_n(x-y)(e^{in\theta}-1)\,d\sigma_y,
 \label{eq:positive-Floquet-remainder}\\
 (R_-^\theta\varphi)(x)
 &:=-\frac1{4\pi}\int_{\partial D}\varphi(y)
   \sum_{n=1}^\infty K_{-n}(x-y)(e^{-in\theta}-1)\,d\sigma_y.
 \label{eq:negative-Floquet-remainder}
\end{align}
To further rewrite the integrand, we set $J_0$ to be the Bessel function of the first kind of order zero and define, for $t>0$,
\begin{equation}
 H_\pm(z;t):=J_0(|z_\perp|t)e^{\pm z_1t}-1.
 \label{eq:definition-H-plus-minus}
\end{equation}
Recalling the Laplace-Bessel identity
\[
 \int_0^\infty e^{-st}J_0(rt)\,dt
 =\frac1{\sqrt{s^2+r^2}},
 \qquad s>0,
\]
we write
\begin{equation}
 K_n(z)=\int_0^\infty H_+(z;t)e^{-nt}\,dt,
 \qquad
 K_{-n}(z)=\int_0^\infty H_-(z;t)e^{-nt}\,dt.
 \label{eq:Laplace-K-plus-minus}
\end{equation}
To sum up the phase factors, let $\rho=e^{-t}$ and define
\begin{equation}
B_\theta(t)
 :=\sum_{n=1}^\infty\rho^n(e^{in\theta}-1)
   =\frac{\rho e^{i\theta}}{1-\rho e^{i\theta}}
    -\frac{\rho}{1-\rho}
   =\frac{\rho}{1-\rho}F_\theta(t),
 \label{eq:definition-B-theta}
\end{equation}
\begin{equation}
F_\theta(t):=\frac{e^{i\theta}-1}
 {1-\rho e^{i\theta}}.
 \label{eq:definition-F-theta}
\end{equation}
Note that the kernels $K_{\pm n}(x-y)$ and their first $x$-derivatives decay as $O(n^{-2})$ when $n\to\infty$, uniformly for $x,y\in\partial D$. Hence, the series defining $R_\pm^\theta$ converge absolutely, and the interchange of the sum with the integrals below is justified by the Fubini theorem. In conclusion, \eqref{eq:positive-Floquet-remainder}-\eqref{eq:definition-B-theta} lead to
\begin{align}
 R_+^\theta\varphi(x)
 &=-\frac1{4\pi}\int_{\partial D}\varphi(y)
   \int_0^\infty H_+(x-y;t)B_\theta(t)\,dt\,d\sigma_y,
 \label{eq:Laplace-R-plus}\\
 R_-^\theta\varphi(x)
 &=-\frac1{4\pi}\int_{\partial D}\varphi(y)
   \int_0^\infty H_-(x-y;t)B_{-\theta}(t)\,dt\,d\sigma_y.
 \label{eq:Laplace-R-minus}
\end{align}
{\color{blue}Step 2.} We next estimate $R_{\pm}^\theta$ using some elementary properties of Bessel functions. Since $D\Subset(0,1)\times\mathbb R^2$, one has
\[
 \gamma_0:=\sup_{x,y\in\partial D}|x_1-y_1|<1.
\]
Fix $\gamma\in(\gamma_0,1)$. The standard bounds for $J_0$ and its derivative, together with the Taylor expansion at $t=0$, imply
\begin{equation}
 |H_\pm(x-y;t)|+|\nabla_xH_\pm(x-y;t)|
 \leq C(1-e^{-t})e^{\gamma t},
 \quad t>0.
 \label{eq:H-plus-minus-bound}
\end{equation}
Here we have enlarged $\gamma_0$ to $\gamma$ to incorporate the factors of $t$ arising when the kernels are differentiated with respect to $x$. For $0<t\leq1$ and $|\theta|\leq1$, note the elementary estimate
\begin{equation}
 |1-e^{-t+i\theta}|^2
 \geq c(t+|\theta|)^2.
 \label{eq:Floquet-denominator-lower-bound}
\end{equation}
On the other hand, direct differentiation of \eqref{eq:definition-F-theta} yields, for $m\geq1$,
\begin{equation}
 |\partial_\theta^mF_\theta(t)|
 \leq C_m\frac{1-e^{-t}}
 {|1-e^{-t+i\theta}|^{m+1}}.
 \label{eq:F-theta-derivative-bound}
\end{equation}
For $m=0$, the same denominator estimate gives
\begin{equation}
 |F_\theta(t)|\leq
 C\frac{|\theta|}{t+|\theta|},
 \quad 0<t\leq1,\quad |\theta|\leq1.
 \label{eq:F-theta-zero-bound}
\end{equation}
For $t\geq1$, the denominator in
\eqref{eq:definition-F-theta} is bounded away from zero, which then leads to
\begin{equation*}
|F_\theta(t)|\leq C|\theta|,\quad |\partial_\theta^mF_\theta(t)|\leq C_m.
\end{equation*}
Substituting these estimates into \eqref{eq:Laplace-R-plus} and using \eqref{eq:Floquet-denominator-lower-bound} gives
\begin{align}
 \|\partial_\theta^mR_\pm^\theta\|
 &\leq  C_m\left(
  \int_0^1\frac{t\,dt}{(t+|\theta|)^{m+1}}+1
  \right) \leq\left\{ \begin{array}{l} 
  C_m(1+|\theta|^{1-m}), \quad m \geq 2,\\
  C (1 + \log (e/|\theta|)), \quad m=1, 
\end{array}
  \right.
 \label{eq:R-plus-higher-estimate}\\
 \|R_\pm^\theta\|
 &\leq C|\theta|\left(
  \int_0^1\frac{dt}{t+|\theta|}+1
  \right) \leq C |\theta| (1+\log(e/|\theta|)) .
 \label{eq:R-plus-zero-estimate}
\end{align}
Note that similar estimates hold for $R_-^\theta$, for which one should apply \eqref{eq:Laplace-R-minus} with $H_-$ and $F_{-\theta}$. Then the proof is complete
\end{proof}

\subsubsection{Schur complement and its derivatives}

The estimates of the capacitance coefficients are obtained by combining Lemma \ref{lem:finite-order-remainder} and the Schur decomposition. The details are presented below. Define
\[
 X:=L^2(\partial D),\qquad Y:=H^1(\partial D),
\]
and recall the orthogonal decompositions
\[
 X=X_\parallel\oplus X_\perp,\qquad
 Y=Y_\parallel\oplus Y_\perp
\]
induced by the (constant-function) projection $P_\parallel$. By Lemma \ref{lem:finite-order-remainder}, we can set $R^0=0$ to continuously extend $R^{\theta}$ throughout the Brillouin zone and write
\[
B_{\theta}:=P_\perp(A+R_{\theta})P_\perp:X_\perp\longrightarrow Y_\perp .
\]
As proved in \cite[Proposition 3.6]{ammari2026longrange}, the operator $B_0$ is invertible.  Since
$\|R_{\theta}\|_{\mathcal B(X,Y)}\to0$ as $\theta\to 0$, by decreasing
$\theta_0$ if necessary, we have
\[
 Q_{\theta}:=B_{\theta}^{-1}:Y_\perp\longrightarrow X_\perp
\]
exists for $0<|\theta|<\theta_0$ and
\[
 \sup_{0<|\theta|<\theta_0}
 \|Q_{\theta}\|_{\mathcal B(Y_\perp,X_\perp)}
 \leq2\|B_0^{-1}\|_{\mathcal B(Y_\perp,X_\perp)}.
\]
The following estimate of arbitrary-order derivatives of $Q_{\theta}$ will be applied later.
\begin{proposition} \label{prop:Q_derivative}
Let
\[
 L(\theta):=\log(e/|\theta|).
\]
For every $m\geq 2$, it holds that
\[
 \|\partial_{\theta}Q_{\theta}\|\leq C L(\theta),\qquad
 \|\partial^{m}_{\theta} Q_{\theta}\|\leq C_m|\theta|^{1-m},
\]
where $C_{m}>0$ depends only on $m$ and the fixed geometry, and not on $\theta$.
\end{proposition}

\begin{proof}
Let $Q_0=B_0^{-1}$. The resolvent identity gives
\[
 Q_{\theta}-Q_0
 =-Q_{\theta}(B_{\theta}-B_0)Q_0,
\]
and hence
\[
 \|Q_{\theta}-Q_0\|
 \leq C|\theta|L(\theta).
\]
Differentiating the identity $B_{\theta}Q_{\theta}=I_{Y_\perp}$ exactly $m$ times and multiplying on the left by $Q_{\theta}$ gives
\begin{equation} \label{eq_Q_derivative_proof_1}
Q_{\theta}^{(m)}
 =-\sum_{k=1}^{m}\binom{m}{k}
 Q_{\theta}B^{(k)}_{\theta}Q_{\theta}^{(m-k)},
 \qquad m\geq1,
\end{equation}
Note that Lemma~\ref{lem:finite-order-remainder} gives
\begin{equation} \label{eq_Q_derivative_proof_2}
\|B^{\prime}_{\theta}\|\leq C L(\theta),\quad
 \|B^{(k)}_{\theta}\|\leq C_k|\theta|^{1-k},
 \quad k\geq2.
\end{equation}
Hence, the preceding identity directly gives that $\|Q^{\prime}_{\theta}\|\leq C L(\theta)$. The estimates of higher order derivatives $Q^{(m)}_{\theta}$ are obtained by induction using \eqref{eq_Q_derivative_proof_1} and \eqref{eq_Q_derivative_proof_2}; the details are omitted here.
\end{proof}

Next, we consider the remaining blocks that appear in the Schur decomposition
\begin{equation*}
a_{\theta}:=P_\parallel (A+R_{\theta})P_\parallel,\quad
b_{\theta}:=P_\parallel (A+R_{\theta})P_\perp,\quad
\widetilde b_{\theta}:=P_\perp (A+R_{\theta})P_\parallel.
\end{equation*}
The Schur correction on the constant subspace $\text{Ran}(P_\parallel)$ is
\[
 G_{\theta}
 :=a_{\theta}-b_{\theta}Q_{\theta}\widetilde b_{\theta}.
\]
The asymptotics of $G_{\theta}$ are recorded as follows.
\begin{proposition} \label{prop:G_derivative}
For every $m\geq 2$, it holds that
\[
 \|G^{\prime}_{\theta}\|\leq C L(\theta),\qquad
 \|G^{(m)}_{\theta}\|\leq C_m|\theta|^{1-m},
 \quad m\geq2,
\]
where $C_{m}>0$ depends only on $m$ and the fixed geometry, and not on $\theta$.
\end{proposition}

\begin{proof}
Since $R_{\theta}\to 0$ as $\theta\to 0$, all the blocks extend continuously to $\theta=0$. Note that
\begin{equation*}
G_{\theta}-G_0=a_{\theta}-a_0-(b_{\theta}-b_0)Q_{\theta}\widetilde b_{\theta}-b_0(Q_{\theta}-Q_0)\widetilde b_{\theta}-b_0Q_0(\widetilde b_{\theta}-\widetilde b_0) .
\end{equation*}
It follows that
\[
 \|G_{\theta}-G_0\|\leq C|\theta|L(\theta),
\]
which proves the estimate of $G^{\prime}_{\theta}$. For the higher order derivatives, the Leibniz rule gives
\[
 G^{(m)}_{\theta}
 =a^{(m)}_{\theta}
 -\sum_{p+q+r=m}\frac{m!}{p!q!r!}
 b^{(p)}_{\theta}Q^{(q)}_{\theta}
 \widetilde b^{(r)}_{\theta}.
\]
Note that i) the zeroth order terms are uniformly bounded, ii) the first order ones are of order $O(L(\theta))$, and iii) the $k$-th order terms ($k\geq 2$) are $O(|\theta|^{1-k})$. Then we conclude the estimate of $G^{(m)}_{\theta}$ by direct calculation.
\end{proof}

The estimate of $G_{\theta}$ presented in Proposition \ref{prop:G_derivative} is transferred to the Floquet symbol of the capacitance operator as follows (see \eqref{eq:exact-capacitance-Schur-complement}). First, note that both $X_\parallel$ and $Y_\parallel$ are spanned by
$\mathbbm 1_{\partial D}$.  Thus, there is a scalar $g(\theta)$ such that
\[
 G_{\theta}[\mathbbm 1_{\partial D}]
 =g(\theta)\mathbbm 1_{\partial D}.
\]
Note that the quasiperiodic single-layer operator is self-adjoint for real $\theta$, and complex conjugation changes $\theta$ into $-\theta$; in particular, the same properties hold for the Schur complement. Therefore, $g(\theta)$ is real and even. The block inverse formula, together with $|\partial D|=1$, gives the following:
\begin{equation} \label{eq:exact-capacitance-Schur-complement}
 c(\theta)=-\frac1{\kappa(\theta)+g(\theta)}.
\end{equation}
By the estimates from Proposition \ref{prop:G_derivative}, we obtain
\begin{equation} \label{eq:g-first-bounds}
|g^{\prime}(\theta)|\leq C\log(e/|\theta|),
\end{equation}
and
\begin{equation} \label{eq:g-higher-bounds}
|g^{(m)}(\theta)|\leq C_m|\theta|^{1-m}\quad (m\geq2).
\end{equation}
Hereafter, for each fixed $m$, the constant $C_m$ is independent of $\theta$. If we further write
\[
 \eta(\theta):=\kappa(\theta)+g(\theta),\quad
 \eta_0(\theta):=\kappa(\theta)+g(0),
\]
then the explicit expression for $\kappa$ in \eqref{eq:kappa-lambda} and the estimates of $g$ in \eqref{eq:g-first-bounds}-\eqref{eq:g-higher-bounds} imply
\[
 |\eta(\theta)|\sim L(\theta),\quad
 |\eta_0(\theta)|\sim L(\theta),
\]
and, for $m\geq1$,
\[
 |\eta^{(m)}(\theta)|+
 |\eta_0^{(m)}(\theta)|
 \leq C_m|\theta|^{-m}.
\]
Consequently,
\begin{equation} \label{eq_eta(theta)_estimate}
\begin{aligned}
|\eta(\theta)^{-1}|+|\eta_0(\theta)^{-1}|
 &\leq\frac{C}{L(\theta)},\\
 \left|\partial_\theta^m\eta(\theta)^{-1}\right|
 +\left|\partial_\theta^m\eta_0(\theta)^{-1}\right|
 &\leq\frac{C_m|\theta|^{-m}}{L(\theta)^2},
 \quad m\geq1.
\end{aligned}
\end{equation}
Now we state the key result leading to the estimate of finite differences of capacitance coefficients, which follows from Lemma \ref{lem:finite-order-remainder}, Proposition \ref{prop:Q_derivative} and \ref{prop:G_derivative}, and demonstrates that the Floquet symbol of capacitance operator admits a complete-Bernstein singular core and a relatively regular remainder. In the sequel, we set
\begin{equation}
 K:=\frac12e^{-4\pi g(0)},
 \quad
 \Phi_K(t):=\frac{4\pi}{\log(1+K/t)}\quad (t>0),
 \quad \Phi_K(0)=0,
 \label{eq:PhiK-definition}
\end{equation}
and separate
\begin{equation}
 \mathcal C_{\mathrm{sing}}
 :=\Phi_K\left(I-\frac{S+S^*}{2}\right),
 \quad
 \mathcal R_{\mathrm{reg}}
 :=\mathcal C-\mathcal C_{\mathrm{sing}}.
 \label{eq:singular-regular-split}
\end{equation}

\begin{proposition}[Singular Hausdorff moments and regular remainder]
\label{prop:core-and-remainder}
The function $\Phi_K$ is complete Bernstein. Define the measure
\[
 d\nu_K(s):=
 \frac{4\pi\,ds}{s(\log^2(K/s-1)+\pi^2)},
 \quad 0<s<K,
\]
and
\[
 r(s)=\frac{\sqrt{s+2}-\sqrt s}{\sqrt{s+2}+\sqrt s},
 \quad
 \sigma_K=\frac{4\pi}{K}\delta_0+
 r_*\left(\frac{1}{s+2}\nu_K\right).
\]
Here, $r_*((s+2)^{-1}\nu_K)$ denotes the push-forward of the measure $(s+2)^{-1}\nu_K$ under the map $s\mapsto r(s)$, and $\delta_0$ is the unit point mass at zero. Then $\sigma_K$ is a finite positive measure on $[0,1)$. Let $A_n:=-c_n^{\mathrm{sing}}$, where $c_n^{\mathrm{sing}}$ are the Fourier coefficients of the operator $\mathcal C_{\mathrm{sing}}$. Then, it holds that
\begin{equation}
 A_n=\frac12\int_{[0,1)}(1-r^2)r^{n-1}\,d\sigma_K(r), \quad n\geq 1.
 \label{eq:singular-Hausdorff-moments}
\end{equation}
On the other hand, denoting the coefficient of $\mathcal R_{\mathrm{reg}}$ by $r_n$, it holds for every fixed $j\geq0$ that
\begin{equation}
 |\Delta^{(j)} r_n|\leq\frac{C_j}{(1+|n|)^{j+2}}.
 \label{eq:regular-remainder-differences}
\end{equation}
\end{proposition}

The proof is postponed to the end of this section. With Proposition \ref{prop:core-and-remainder}, the estimate of finite differences of capacitance coefficients follows directly:

\begin{proposition}[Finite-difference estimate of the capacitance coefficients]
\label{prop:verify-critical-differences}
For $j=0,1,2$ and $n\neq0$,
\begin{equation}
 |\Delta^{(j)} c_n|\leq
 \frac{C_j}{(1+|n|)^{j+1}\log^2(e+|n|)}.
 \label{eq:capacitance-finite-difference-estimates}
\end{equation}
Consequently, the scalar capacitance operator satisfies
Assumptions~\ref{asmp_decay} and \ref{asmp_finite-difference} with
$\beta=2$.
\end{proposition}

\begin{proof}
The standard Abelian estimate for the explicit symbol
$\Phi_K(1-\cos\theta)$ gives
\begin{equation}
 0\leq A_n\leq
 \frac{C}{(1+n)\log^2(e+n)},
 \label{eq:An-Abelian-bound}
\end{equation}
which indeed follows by integrating the Fourier coefficient once and using
\[
 \Phi_K^{\prime}(t)=
 \frac{4\pi K}{t(t+K)\log^2(1+K/t)}.
\]
The details are skipped here. Now, let $m=\lfloor n/2\rfloor$. The identity \eqref{eq:singular-Hausdorff-moments} gives
\begin{align*}
 |\Delta A_n|
 &\leq
 \sup_{0\leq r\leq1}(1-r)r^{n-m}A_m
 \leq\frac{C}{n}A_m
 \leq\frac{C}{n^2\log^2(e+n)},\\
 0\leq\Delta^{(2)} A_n
 &\leq
 \sup_{0\leq r\leq1}(1-r)^2r^{n-m}A_m
 \leq\frac{C}{n^2}A_m
 \leq\frac{C}{n^3\log^2(e+n)}.
\end{align*}
Together with \eqref{eq:regular-remainder-differences}, these estimates
prove \eqref{eq:capacitance-finite-difference-estimates} for $n>0$.
Reality and evenness of $c_n$ give the desired estimates for negative indices.
\end{proof}

Finally, we present the proof of Proposition \ref{prop:core-and-remainder}.

\begin{proof}[Proof of Proposition \ref{prop:core-and-remainder}]
{\color{blue}Step 1.} For $z$ in the upper half-plane, $1+K/z$ lies in the lower half-plane. The principal logarithm therefore has negative imaginary part and its reciprocal has positive imaginary part. Noting that $\Phi_K$ is positive on $(0,\infty)$, we know by the Pick characterization that $\Phi_K$ is complete Bernstein. For $0<s<K$, its boundary value on the cut is
\[
 \Phi_K(-s+i0)
 =\frac{4\pi}{\log(K/s-1)-i\pi},
\]
whereas its boundary imaginary part vanishes for $s>K$.  Moreover,
\[
 \Phi_K(0^{+})=0,
 \quad
 \lim_{t\to\infty}\frac{\Phi_K(t)}{t}=\frac{4\pi}{K}.
\]
The Stieltjes inversion formula for complete Bernstein functions therefore gives the explicit representation
\begin{equation}
 \Phi_K(t)=\frac{4\pi}{K}t
 +\int_0^K\frac{t}{t+s}
 \frac{4\pi\,ds}{s(\log^2(K/s-1)+\pi^2)}.
 \label{eq:PhiK-integral-representation}
\end{equation}
The endpoint behavior of the density in \eqref{eq:PhiK-integral-representation} is
\[
 \frac{d\nu_K}{ds}(s)
 =O\left(\frac{1}{s\log^2(K/s)}\right)
 \quad(s\to 0^{+}),
 \quad
 \frac{d\nu_K}{ds}(s)
 =O\left(\frac{1}{\log^2(K/(K-s))}\right)
 \quad(s\to K^{-}).
\]
In particular, the first majorant is integrable and the second one is bounded near $K$. Hence $\nu_K$ is finite. The same endpoint estimates, multiplied by the bounded factor $(s+2)^{-1}$, show that the continuous part of $\sigma_K$ is finite; its remaining atom has mass $4\pi/K$.

{\color{blue}Step 2.} Now we estimate the Fourier coefficient of the singular part. First, for $s>0$, the Poisson expansion gives
\[
 \frac{1}{s+1-\cos\theta}
 =\frac{1}{\sqrt{s(s+2)}}
 \left(1+2\sum_{m=1}^{\infty}r(s)^m\cos(m\theta)\right).
\]
Since $\nu_K$ is finite, the interchange of Fourier integral and that in \eqref{eq:PhiK-integral-representation} is justified by the dominated convergence theorem. For every $n\geq1$, this gives the following formula of the corresponding Fourier coefficients
\begin{equation} \label{eq_coeff_An}
A_n=\frac{2\pi}{K}\mathbbm 1_{\{n=1\}}
 +\int_0^K\frac{s}{\sqrt{s(s+2)}}r(s)^n\,d\nu_K(s).
\end{equation}
Note the elementary identity
\[
 \frac{s}{\sqrt{s(s+2)}}r(s)^n
 =\frac{1-r(s)^2}{2(s+2)}r(s)^{n-1}.
\]
The continuous term therefore equals the contribution of $r_*((s+2)^{-1}\nu_K)$ to \eqref{eq:singular-Hausdorff-moments}.  The Fourier coefficient of $(4\pi/K)(1-\cos\theta)$ at $n=1$ is $-2\pi/K$, and the atom $(4\pi/K)\delta_0$ contributes $(1/2)(4\pi/K)=2\pi/K$ to $A_1$ and zero to $A_n$ for $n\geq2$. This proves \eqref{eq:singular-Hausdorff-moments}.

{\color{blue}Step 3.} To estimate the regular remainder, let
\[
 \phi_0(\theta):=-\frac1{\kappa(\theta)+g(0)}
 =\frac{4\pi}{\log(K/(1-\cos\theta))}
\]
for sufficiently small $|\theta|$. We first estimate the difference $c(\theta)-\phi_0(\theta)$ and then the difference between $\phi_0(\theta)$ and the complete-Bernstein core, showing that both parts admit the required estimates. For the first part, note that \eqref{eq:exact-capacitance-Schur-complement} implies the identity
\[
 c(\theta)-\phi_0(\theta)
 =(g(\theta)-g(0))W(\theta),
 \quad
 W(\theta)=
 \frac1{(\kappa(\theta)+g(\theta))
          (\kappa(\theta)+g(0))}.
\]
By \eqref{eq:g-first-bounds},\eqref{eq:g-higher-bounds} and \eqref{eq_eta(theta)_estimate}, one can calculate that
\[
 |W(\theta)|\leq\frac{C}{L(\theta)^2},\quad
 |W^{(m)}(\theta)|
 \leq\frac{C_m|\theta|^{-m}}{L(\theta)^3},
 \quad (m\geq 2).
\]
The same estimate for $W^{\prime}$ holds when $m=1$. Applying the Leibniz rule to $(g(\theta)-g(0))W(\theta)$ gives
\begin{equation*}
|c(\theta)-\phi_0(\theta)| \leq\frac{C|\theta|}{L(\theta)},\quad
|\partial_\theta(c(\theta)-\phi_0(\theta))|\leq\frac{C}{L(\theta)},
\end{equation*}
and
\[
|\partial_\theta^m( c(\theta)-\phi_0(\theta) )|
\leq \frac{C_m|\theta|^{1-m}}{L(\theta)^2} \quad (m\geq 2).
\]
For $m\geq2$, note that i) the term on which all derivatives fall on $g-g(0)$ has the required bound, ii) the term on which all derivatives fall on $W$ has the same bound because $|g-g(0)|\leq C|\theta|L(\theta)$; every mixed term has an additional negative power of $L(\theta)$.

It remains to compare $\phi_0$ with the complete-Bernstein core.  Write
\[
 \ell_K(\theta):=\log(K/(1-\cos\theta)),\quad
 \delta_K(\theta):=\log(1+(1-\cos\theta)/K).
\]
After decreasing $\theta_0$ if necessary, one has the asymptotics $\ell_K(\theta)\sim L(\theta)$, and
\[
 \log(1+K/(1-\cos\theta))
 =\ell_K(\theta)+\delta_K(\theta).
\]
Note that the function $\delta_K$ is smooth and vanishes to second order at zero. Thus, for every fixed $m\geq 0$ and $0<|\theta|<\theta_0$, it holds
\[
 |\delta_K^{(m)}(\theta)|
 \leq C_m|\theta|^{2-m}.
\]
Consequently,
\[
 \begin{split}
 \left|
 \big(\ell_K(\theta)(\ell_K(\theta)+\delta_K(\theta))\big)^{-1}
 \right|
 &\leq\frac{C}{L(\theta)^2},\\
 \left|
 \partial_\theta^m
 \big(\ell_K(\theta)(\ell_K(\theta)+\delta_K(\theta))\big)^{-1}
 \right|
 &\leq
 \frac{C_m|\theta|^{-m}}{L(\theta)^3}
 \quad (m\geq1).
 \end{split}
\]
Since
\[
 \phi_0(\theta)-\Phi_K(1-\cos\theta)
 =
 \frac{4\pi\delta_K(\theta)}
 {\ell_K(\theta)
  (\ell_K(\theta)+\delta_K(\theta))},
\]
a direct calculation yields
\[
 \left|
 \partial_\theta^m
 (\phi_0(\theta)-\Phi_K(1-\cos\theta))
 \right|
 \leq
 \frac{C_m|\theta|^{2-m}}{L(\theta)^2}.
\]
This completes the estimate of the remainder. In summary, the symbol $q_{\mathrm{reg}}$ of $\mathcal R_{\mathrm{reg}}$ satisfies
\begin{equation} \label{eq_complete_bernstein_1}
|q_{\mathrm{reg}}(\theta)|\leq C|\theta|/L(\theta),\quad
|q_{\mathrm{reg}}^{\prime}(\theta)|\leq C/L(\theta),\quad
|q_{\mathrm{reg}}^{(m)}(\theta)|\leq C_m|\theta|^{1-m}/L(\theta)^2 \quad (m\geq 2).
\end{equation}
Finally, noting that
\[
 \Delta^{(j)} r_n
 =\frac1{2\pi}\int_{-\pi}^{\pi}
(e^{-i\theta}-1)^jq_{\mathrm{reg}}(\theta)e^{-in\theta}\,d\theta,
\]
the above estimates together with the smoothness of the integrand away from zero imply that
\[
 |\Delta^{(j)} r_n|
 \leq\frac{C_j}{(1+|n|)^{j+2}},
\]
which is proved by applying the integration by parts. This completes the proof of \eqref{eq:regular-remainder-differences}.
\end{proof}




\subsection{Local and global results for one-resonator chains}
\label{sec:scalar-regularity-verification}

We now verify the local Floquet hypothesis of Proposition~\ref{prop:local-no-ac}. The argument establishes strict monotonicity near the acoustic endpoint and excludes other preimages at sufficiently small positive energies. After that, we verify the global monotonicity in Assumption~\ref{asmp_eigenvalue_monotone} when the reference resonator has an additional reflection symmetry.

\begin{proposition}[Regular acoustic window]
\label{prop:capacitance-regular-acoustic-window}
There is $E_{\mathrm{reg}}>0$ such that, for every
$E\in(0,E_{\mathrm{reg}})$, the complete inverse image of the scalar
capacitance symbol is
\[
 c^{-1}(\{E\})=\{-\theta(E),\theta(E)\},\quad
 0<\theta(E)<\pi,
\]
where $c$ is $C^2$ near both points and $\partial_{\theta}c(\theta(E))>0$. In particular, on every compact interval $\mathcal I\Subset(0,E_{\mathrm{reg}})$ these preimages remain a positive distance from $0$ and $\pm\pi$, and $|\partial_{\theta}c|$ has a positive lower bound on $c^{-1}(\mathcal I)$.
\end{proposition}

\begin{proof}
By \eqref{eq:singular-regular-split}, the scalar symbol has the form
\begin{equation}
 c(\theta)=\Phi_K(1-\cos\theta)+q_{\mathrm{reg}}(\theta).
 \label{eq:capacitance-symbol-core-plus-remainder}
\end{equation}
The derivative formula following from \eqref{eq:PhiK-integral-representation} is
\begin{equation}
 \Phi_K^{\prime}(t)=
 \frac{4\pi K}{t(t+K)\log^2(1+K/t)}.
 \label{eq:PhiK-derivative-for-monotonicity}
\end{equation}
Since $1-\cos\theta\sim\theta^2/2$ and
$\sin\theta\sim\theta$ as $\theta\to 0$, one has
\begin{equation}
 \Phi_K^{\prime}(1-\cos\theta)\sin\theta
 \sim\frac{1}{\theta\log^2(e/\theta)}.
 \label{eq:singular-core-positive-derivative}
\end{equation}
On the other hand, Proposition~\ref{prop:core-and-remainder} gives
\begin{equation}
 |q_{\mathrm{reg}}^{\prime}(\theta)|\leq
 \frac{C}{\log(e/\theta)}.
 \label{eq:regular-remainder-small-derivative}
\end{equation}
The positive quantity in \eqref{eq:singular-core-positive-derivative} dominates \eqref{eq:regular-remainder-small-derivative}. Hence, there is $\theta_0\in(0,\pi)$ such that
\begin{equation}
 c^{\prime}(\theta)>0,
 \qquad0<\theta<\theta_0.
 \label{eq:capacitance-low-energy-monotonicity}
\end{equation}
The Dirichlet principle gives $c(\theta)>0$ for $\theta\neq0$ (cf. \cite{cbms,ammari2026longrange}). Therefore,
\[
 m_0:=\min_{\theta_0\leq|\theta|\leq\pi}c(\theta)>0.
\]
Choose $E_{\mathrm{reg}}<\min\{m_0,c(\theta_0)\}$. Since $c(0)=0$, and $c$ is strictly increasing on $(0,\theta_0)$, each $E\in(0,E_{\mathrm{reg}})$ has exactly one positive preimage $\theta(E)\in(0,\theta_0)$. Evenness of $c(\theta)$ gives the negative preimage, and the definition of $m_0$ excludes all other preimages. For compact $\mathcal I\Subset(0,E_{\mathrm{reg}})$, the continuity of the inverse of the derivative $c'$ gives the stated uniform bounds.
\end{proof}


By the verifications of assumptions performed in this section, the absence of absolutely continuous spectrum of the capacitance operator under the geometric hypotheses and the random law assumptions stated in Section \ref{sec_main_result} follows directly from our abstract result Proposition \ref{prop:local-no-ac}:

\begin{corollary}[Absence of absolutely continuous spectrum on the acoustic window]
\label{cor:scalar-all-assumptions}
Let $D\Subset(-a/2,a/2)\times\mathbb R^2$ be connected with $C^{2,\eta}$ boundary, and suppose that there is one resonator per unit cell. Then there are $\epsilon_0>0$ such that, for every $0<\epsilon<\epsilon_0$, $\mathcal I\Subset(0,E_{\mathrm{reg}})$ and almost every realization, it holds that
\begin{equation}
 \mathbbm 1_{\mathcal I}(\mathcal C_{\omega,\epsilon})
 P_{\mathrm{ac}}(\mathcal C_{\omega,\epsilon})=0.
 \label{eq:scalar-final-no-ac}
\end{equation}
\end{corollary}

As we have indicated before, the reflection symmetry of the reference resonator gives rise to the global monotonicity of the Floquet symbol of the capacitance operator. To be precise, the reflection operation is defined as $\mathcal R(x_1,x_\perp):=(-x_1,x_\perp)$, and we assume
\[
 \mathcal R D=D,
\]
with $D\Subset(-a/2,a/2)\times\mathbb R^2$ being the reference resonator.

\begin{proposition}[Global monotonicity under reflection symmetry]
\label{prop:reflection-global-monotonicity}
Assume that $\partial D$ is $C^{2,\eta}$ and $\mathcal R D=D$. Then, it holds that
\begin{equation}
 c'(\theta)>0,\quad 0<\theta<\pi.
 \label{eq:reflection-global-positive-derivative}
\end{equation}
In particular, $c$ is strictly increasing on $(0,\pi)$ and satisfies
Assumption~\ref{asmp_eigenvalue_monotone}.
\end{proposition}

The proof is given in Appendix~\ref{app:reflection-monotonicity}, which is based on the application of the Dirichlet principle. Now we immediately conclude that the following result holds. In fact, Proposition \ref{prop:verify-critical-differences} verifies Assumptions \ref{asmp_decay} and \ref{asmp_finite-difference} with $\beta=2$; Proposition \ref{prop:reflection-global-monotonicity} and \eqref{eq:basic-symbol-properties}, with continuity stated there, verify Assumption \ref{asmp_eigenvalue_monotone}. Theorem \ref{eq:smoothed-cut-limit} gives the conclusion.

\begin{corollary}[Global absence of absolutely continuous spectrum for reflection-symmetric chains]
\label{cor:reflection-global-no-ac}
Suppose the assumptions in Corollary~\ref{cor:scalar-all-assumptions} hold and additionally that $\mathcal R D=D$. Then there exists $\epsilon_0>0$ such that, for every $0<\epsilon<\epsilon_0$,
\begin{equation}
 P_{\mathrm{ac}}(\mathcal C_{\omega,\epsilon})=0
 \quad\text{almost surely}.
 \label{eq:reflection-global-no-ac}
\end{equation}
\end{corollary}




\section{Concluding remarks} \label{sec:conclusion}
In this paper, we have proved the absence of absolute continuous spectrum for critical long-range lattice operators under multiplicative i.i.d. perturbation, assuming a monotonicity condition on the Floquet symbol and a finite-difference estimate. All the assumptions are verified on a concrete physical model, the capacitance operator associated with the 1D-in-3D resonator chain. These results provide an analytical support for the numerical localization phenomena reported in~\cite{ammari2026longrange}.

There are several questions that remain open. In particular, the absence of absolutely continuous spectrum does not distinguish the pure point spectrum from the singular continuous spectrum. Proving pure-point localization at critical decay requires additional control of the singular spectral component, which appears to be hard. It will be interesting to see whether additional structural information of the capacitance operator helps to prove the pure-point localization. On the other hand, we point out that the global absence of absolute continuous spectrum is only verified under an additional reflection symmetry on the resonators; extending this global conclusion to geometries without the stated reflection symmetry requires further information about the symbol. This is left for future study.

\bigskip
\textbf{Conflict of Interest Statement}: On behalf of all authors, the
corresponding author states that there is no conflict of interest.

\bigskip
\textbf{Data Availability Statement}: On behalf of all authors, the
corresponding author confirms that the current paper does not have
associated data.

\bigskip
\textbf{AI Usage}: The authors first realized that the sparse-gap construction works for our critical long-range lattice operators, and then asked ChatGPT 5.6 Pro to exploit additional structural information of the capacitance operator by which a trace-class cut estimate is available; it pointed out the finite-difference structure and formulated the proof. The authors then substantially simplified the arguments and rewrote the manuscript.

\appendix
\section{Global monotonicity from reflection symmetry}
\label{app:reflection-monotonicity}

In this appendix, we prove Proposition~\ref{prop:reflection-global-monotonicity} using the quasiperiodic Dirichlet principle. For convenience, we write $L=a/2$ and 
\[
Y_c:=(-L,L)\times\mathbb R^2,\quad \Pi:=(0,L)\times\mathbb R^2,\quad D_+:=D\cap\Pi.
\]
The reference domain satisfies $D\Subset Y_c$ and the reflection symmetry $\mathcal R D=D$, where $\mathcal R(x_1,x_\perp)=(-x_1,x_\perp)$. For a function on $Y_c$, denote its traces onto the face $x_1=\pm L$ by $\gamma_\pm$, respectively. Similarly, for functions on $\Pi$, write $\gamma_0$ and $\gamma_L$ for the traces at $x_1=0$ and $x_1=L$.

\begin{proof}[Proof of Proposition~\ref{prop:reflection-global-monotonicity}]
{\color{blue}Step 1.} We first note that the capacitance equals the minimum of Dirichlet energy under the quasi-periodic boundary condition. To be precise, fix $0<\theta<\pi$ and define
\begin{equation}
 \mathscr U_\theta=
 \left\{u\in H^1(Y_c;\mathbb C):
 u=1\ \text{a.e. in }D,\quad
 \gamma_+u=e^{i\theta}\gamma_-u\right\}.
 \label{eq:reflection-full-admissible-space}
\end{equation}
Then, as proved in \cite{cbms}, it holds that
\begin{equation}
 c(\theta)=\min_{u\in\mathscr U_\theta}
 \int_{Y_c}|\nabla u|^2\,dx,
 \label{eq:reflection-dirichlet-principle}
\end{equation}
recalling that $c(\theta)$ is the Floquet symbol of the capacitance operator. Moreover, the coercivity and strict convexity give a unique minimizer $u_\theta$; see \cite{cbms} for a detailed discussion.

{\color{blue}Step 2.} We next represent the symbol $c(\theta)$ in a convenient form by which the continuity in $\theta$ is clearer; see \eqref{eq:reflection-fixed-space-minimum}. First, we note that the anti-linear map
\[
 (\mathscr T u)(x)=\overline{u(\mathcal R x)}
\]
is invariant on the quasi-periodic function space $\mathscr U_\theta$ and preserves the Dirichlet energy in \eqref{eq:reflection-dirichlet-principle}. Consequently, the uniqueness of the minimizer $u_\theta$ implies
\begin{equation}
 u_\theta(\mathcal R x)=\overline{u_\theta(x)}.
 \label{eq:reflection-conjugation-symmetry}
\end{equation}
Thus, its real part is even in $x_1$ and its imaginary part is odd. On the face $x_1=L$, the two trace conditions imply
\[
 \gamma_Lu_\theta=e^{i\theta}\overline{\gamma_Lu_\theta}.
\]
Writing $u_\theta=v+iw$, this is equivalent to
\begin{equation}
 \gamma_Lw=\tan(\theta/2)\gamma_Lv.
 \label{eq:reflection-face-relation}
\end{equation}
In addition, $v=1$ and $w=0$ in $D_+$, and $\gamma_0w=0$. Hence, if we write $w:=tz$ with $t=\tan(\theta/2)>0$, it holds that
\begin{equation}
 \begin{split}
 \mathscr A=\{(v,z)\in H^1(\Pi;\mathbb R)^2:
 &\ v=1,\ z=0\ \text{a.e. in }D_+,\\
 &\ \gamma_0z=0,\quad\gamma_Lz=\gamma_Lv\}.
 \end{split}
 \label{eq:reflection-half-admissible-space}
\end{equation}
We remark that the set $\mathscr A$ is independent of
$\theta$. Conversely, every pair $(v,z)$ in $\mathscr A$ gives an admissible function in $\mathscr U_\theta$ by setting $u:=v+itz$ on $\Pi$ and extending it through $u(-x_1,x_\perp)=v(x_1,x_\perp)-itz(x_1,x_\perp)$; the verification of boundary conditions are direct. In summary, the Floquet symbol $c(\theta)$ admits the following characterization:
\begin{equation} \label{eq:reflection-fixed-space-minimum}
c(\theta)=F(s),\quad s=\tan^2(\theta/2),
\end{equation}
where
\begin{equation*}
F(s)=2\min_{(v,z)\in\mathscr A}
 \left\{A(v)+sA(z)\right\},\quad s>0,
\end{equation*}
and
\begin{equation*}
A(v):=\int_\Pi|\nabla v|^2\,dx .
\end{equation*}

{\color{blue}Step 3.} We claim that, for any $s>0$, the minimizer $(v_{s},z_{s})$ satisfies
\begin{equation}
 B(z_s)>0,
 \label{eq:reflection-strict-energy}
\end{equation}
i.e., the second energy for the minimizer is strictly positive. Indeed, if $B(z_s)=0$, then $z_s$ is constant on the connected strip $\Pi$; its zero trace at $x_1=0$ makes it identically zero. The associated full-cell minimizer $u_{\theta}$ is therefore real and even. However, the quasi-oeriodic boundary condition leads to
\[
 (1-e^{i\theta})\gamma_+u_\theta=0,
\]
hence $\gamma_+u_\theta=0$; similarly, one can show the normal derivative also vanishes in the sense of trace. Consequently, the unique extension property of harmonic function indicates that $u_\theta$ vanishes in the exterior component adjacent to the boundary. However, that component has a nonempty portion of $\partial D$ in its boundary: one can follow a line parallel to the chain from the face toward a rightmost exposed portion of $D$. This contradicts the boundary trace $u_\theta=1$ there, proving \eqref{eq:reflection-strict-energy}.

{\color{blue}Step 4.} We now prove the desired monotonicity of Floquet symbol. Suppose $0<s_1<s_2$. We know by the result of Step 2 that
\begin{equation}
 2(s_2-s_1)B(z_{s_2})
 \leq F(s_2)-F(s_1)
 \leq2(s_2-s_1)B(z_{s_1}).
 \label{eq:reflection-comparison-inequality}
\end{equation}
Then, using the result of Step 3, we directly arrives at
\[
 F'(s)=2B(z_s)>0,
\]
where the continuity of $B(z_s)$ is applied (the details are omitted here). This completes the proof of \eqref{eq:reflection-global-positive-derivative}.
\end{proof}

\footnotesize
\bibliographystyle{plain}
\bibliography{ref}

\end{document}